\documentclass[aps,pra,superscriptaddress,longbibliography,amsmath,amssymb,twocolumn]{revtex4-1} 
\usepackage{graphicx,natbib,physics,fontenc,amsthm,textcomp,color,bbm,hyperref}
\hypersetup{
    colorlinks = true,
    citecolor=blue,
    linkcolor = red,
    anchorcolor = red,
    citecolor = blue,
    filecolor = red,
    pagecolor = red,
    urlcolor = blue,
    }
\usepackage{inputenc}
\newtheorem{Theorem}{Theorem}

\begin{document}

\title{Measure of invertible dynamical maps under convex combinations of noninvertible dynamical maps}
	\author{Vinayak Jagadish}
	\email{vinayak.jagadish@uj.edu.pl}
	\affiliation{Institute of Theoretical Physics, Jagiellonian University, ul. {\L}ojasiewicza 11, 30-348 Krak\'ow, Poland}
				\author{R. Srikanth}
\affiliation{Theoretical Sciences Division,
Poornaprajna Institute of Scientific Research (PPISR), 
Bidalur post, Devanahalli, Bengaluru 562164, India}
	\author{Francesco Petruccione}
	\affiliation{School for Data Science and Computational Thinking, Stellenbosch University, Stellenbosch 7600, South Africa}
	\affiliation{National Institute for Theoretical and Computational Sciences (NITheCS),  Stellenbosch 7600, South Africa}
	\affiliation{Quantum Research Group, School  of Chemistry and Physics,
		University of KwaZulu-Natal, Durban 4001, South Africa}

\begin{abstract} 

We study the convex combinations of the $(d+1)$ generalized Pauli dynamical maps in a Hilbert space of dimension $d$. For certain choices of the decoherence function, the maps are noninvertible and they remain under convex combinations as well. For the case of dynamical maps characterized by the decoherence function $(1-e^{-ct})/n$ with the decoherence parameter $n$ and decay factor $c$, we evaluate the fraction of invertible maps obtained upon mixing, which is found to increase superexponentially with dimension $d$. 

\end{abstract}
\maketitle  
\section{Introduction} 

For a quantum system undergoing an open evolution~\cite{petruccione}, the time evolution of the system of interest is described by a quantum dynamical map which is a time-continuous family of completely positive (CP) and trace-preserving  linear time-dependent maps, $\{\Phi (t): t\geq 0, \Phi(0) = \mathbbm{1}\}$, acting on the bounded operators of the Hilbert space of the system of interest~\cite{sudarshan_stochastic_1961, Quanta77}.
The properties of the dynamical map is related to the time-local generator $L(t)$~\cite{gorini_completely_1976} in the time-local master equation, $\dot{\Phi}(t) = L(t)\Phi(t)$, with 
\begin{equation}
\begin{split}
L(t)[\rho]=&-\imath [H(t),\rho]\\&
+\sum_i \gamma_i (t)
\left(G_i(t)\rho G_i(t)^\dagger-\frac {1}{2}\{G_i(t)^\dagger G_i(t),\rho\}\right).
\label{gksl}
\end{split}
\end{equation}
Here $H(t)$ is the effective Hamiltonian, $G_i(t)$'s are the noise operators, and $\gamma_i (t)$'s the decoherence rates, describing the continuous-in-time description of the dynamics of the system of interest.
The dynamical map is divisible if
\begin{equation}
\Phi(t_f, t_i) = K(t_f, t) \Phi(t, t_i),\quad \forall t_f\geq t \geq t_i\geq 0.
\label{cpdivdef}
\end{equation}
The map is CP divisible if for all $t$'s, the propagator $K(t_f, t)$ is CP. If the decay rates $\gamma_i (t)$ are positive at all times, the map is CP divisible.  Otherwise, the map is CP indivisible. According to the CP divisibility criterion~\cite{rivas_entanglement_2010} of quantum Markovianity, a quantum dynamical map is Markovian if it is CP divisible at all instants of time. Any deviation from CP divisibility is, therefore, an indicator of non-Markovianity. This is also equivalent to temporary negativity of the decay rates in the time-local master equation~\cite{hall2010}. If the decay rates are constants, the generator corresponds to a Markovian semigroup. Such maps are time homogenous, i.e., 
\begin{equation}
\Phi(t_f, t_i) = \Phi(t_f - t_i, 0) :=  \Phi(t_f - t_i) .
\end{equation}

Convex combinations of quantum dynamical maps have been studied by various authors. An example to create a non-Markovian evolution by mixing two Markovian semigroups was reported in Ref.~\cite{chruscinski_non-markovianity_2015}. A similar idea appeared in an earlier paper~\cite{wolf_assessing_2008} prior to the advent of the divisibility and distinguishability criteria for quantum non-Markovianity. In Ref.~\cite{breuer_mixing-induced_2018}, the counterintuitive behavior was explained in terms of information flow between the system and the environment. An example of a convex combination of two non-Markovian dynamical maps leading to a Markovian evolution was discussed in Ref.~\cite{wudarski2016}. It was shown that an eternally non-Markovian evolution (with one decay rate being negative at all times, $t>0$) results from a convex combination of Pauli semigroups~\cite{megier_eternal_2017}. More generally, eternally non-Markovian evolution obtained by the convex combination of semigroups of generalized Pauli dynamical maps has been studied in Ref.~\cite{siudzinska_jpa_2020}.
Recently, the geometry of the region of non-Markovianity obtained by convex combinations of Pauli semigroups and time-dependent Markovian Pauli dynamical maps was addressed~\cite{jagadish_convex_2020,jagadish_nonqds_2020}. These results showed that the sets of CP divisible and CP indivisible Pauli dynamical maps are nonconvex. A measure of non-Markovianity was introduced~\cite{desantis2020} based on the minimal amount of extra Markovian noise that needs to be added via convex mixing in order to make the dynamical map Markovian.

In Eq. (\ref{cpdivdef}), if for a particular instant of time $t = t^{*}$, the map $\Phi(t^{*},t_i)$ is noninvertible, then there are multiple initial states that are mapped to the same output state. In this case, the propagator $K(t_f, t^{*}) =  \Phi(t_f, t_i) \Phi(t^{*} t_i )^{-1}$ is undefined. Such points $t^{*}$ are referred to as singular points~\cite{hou_singularity_2012} of the map. The presence of singularities in the map is reflected in the corresponding  time-local master equation where one or more of the decay rates could go to infinity. Despite the fact that the dynamical map could be momentarily noninvertible, the singularities are usually not pathological, in the sense that the density operator of the system of interest is well defined~\cite{jagadish_measure2_2019, jagadish_initial_2021}. Recently, convex combinations of dynamical maps with singularities have caught attention~\cite{utagi_singularities_2021,siudzinska_markovian_2021,jagadish2022noninvertibility,siudzinska_jpa_2022}. For generalized Pauli dynamical maps, it was shown that noninvertible maps cannot be produced by mixing invertible maps. The surprising result that convex combinations of noninvertible qudit Pauli maps can produce a semigroup was reported in Ref.~\cite{siudzinska_markovian_2021}. Subsequently, it was also shown that noninvertibility of the generalized Pauli input maps is necessary for this~\cite{jagadish2022noninvertibility}.  This prompts the question of whether the noninvertibility is sufficient. The answer is no: Some mixtures produce noninvertible maps. A natural followup to this line of inquiry would be to quantify the fraction of the mixtures of noninvertible maps that produce (non)invertible maps.

The physical motivation is ultimately about fighting noise in practical quantum information processing, especially in the context of near-term quantum devices. In principle, noise corresponding to a noninvertible map cannot be corrected. An experimentalist would want to avoid parameter regimes that would generate such noise. In a futuristic quantum computer that is fully fault tolerant and operating within the required noise thresholds, this is less of a problem. However, as we enter the noisy intermediate scale quantum era~\cite{preskill_quantum_2018}, where quantum devices are expected to lack the depth required to implement full quantum error correction, avoidance of noninvertibility can be a matter of practical concern. 
\color{black}

The paper is organized as follows.  We discuss the case of mixing noninvertible Pauli dynamical maps in the following Section. In Sec.~\ref{genpaulisec}, we extend our analysis to the case of generalized Pauli dynamical maps. Finally, we discuss the results and conclude in Sec.~\ref{conc}. 

\section{Convex Combination of noninvertible Pauli dynamical maps}
\subsection{Obtaining invertible maps and measure of the corresponding set}
\label{paulisec}
Consider the three Pauli dynamical maps, 
\begin{equation}
\label{paulichanndef}
\Phi_i (t)[\rho]=[1-p(t)]\rho + p(t)\sigma_i\rho \sigma_i, i= 1,2,3
\end{equation} 
with $p(t)$ being the decoherence function and $\sigma_i$ the Pauli matrices. The convex combination of the three Pauli dynamical maps Eq.~(\ref{paulichanndef}), each mixed in proportions of $x_i$ is, therefore,
\begin{equation}
\label{outputmappauli}
\tilde{\Phi}(t) = \sum_{i=1}^{3} x_{i} \Phi_i (t),  \quad (x_i >0, \sum_i x_i =1).
\end{equation}
We will call the three $\Phi_i (t)$'s input maps and $\tilde{\Phi}(t)$ the output map. Let each of the maps be characterized by the same decoherence function, 
\begin{equation}
\label{decohfunc}
p(t) = \frac{1-e^{-ct}}{n}, \quad n\geq 1, c >0.
\end{equation}
 From the corresponding time-local master equation for each of the input maps $\Phi_i$,
\begin{equation}
\label{megen}
L_{i}(t)[\rho] = \gamma (t) (\sigma_i\rho\sigma_i-\rho),
\end{equation}
with the decay rate
\begin{equation}
\gamma (t) = \frac{c}{(n-2)e^{c t}+2},
\end{equation}
it can be seen that for $n \ge 2$, the decay rates never become infinite at finite times $t$, and therefore, the input maps are all invertible. The output map $\tilde{\Phi}$ [Eq. (\ref{outputmappauli} ] satisfies the eigenvalue relation,
\begin{equation}
\label{eigpauli}
\tilde{\Phi}(t)[\sigma_i] = \lambda_i (t) \sigma_i,
\end{equation}
with
\begin{equation}
\lambda_i(t) =1- 2(1-x_i)p(t),
\label{eq:lambdaz}
\end{equation} 
with $p(t)$ given by Eq. (\ref{decohfunc}).
The eigenvalues $\lambda_i$ becomes singular at
\begin{equation}
t^{*}=\frac{1}{c}\ln \Big[\frac{2(1-x_i)}{2(1-x_i)-n}\Big].
\label{eq:tstar}
\end{equation}
From Eq. (\ref{eq:tstar}), the requirement for invertibility is
\begin{equation}
x_i \ge 1-\frac{n}{2}.
\label{eq:xi2}
\end{equation}
For invertible inputs, $n\geq2$, the right-hand side above is negative and, thus, any mixing coefficient $x_i$ satisfies the lower bound in Eq. (\ref{eq:xi2}). This bound monotonically increases as $n$ decreases in Eq. (\ref{eq:xi2}). If $n < \frac{4}{3}$, then $x_i < \frac{1}{3}$, a constraint that cannot be satisfied by all three maps \big(since in that case we would have $\sum_i x_i < 1$ \big). We now have the following result.

When the input map is invertible ($n\ge2$), the output map is necessarily invertible. For noninvertible inputs, if $n< \frac{4}{3}$, the output map is necessarily noninvertible. For input maps that are of the intermediate noninvertible range \big($2 > n > \frac{4}{3}$ \big), the output can be invertible or not, depending on the mixing coefficients. In the intermediate range of noninvertible inputs $n \in (\frac{4}{3},2)$, a fraction of the output maps will be invertible.  In the case of $n=4/3$, for the output map to be invertible, Eq. (\ref{eq:xi2}) implies that all $x_i \ge \frac{1}{3}$. Therefore, the output map is noninvertible for all values of mixing coefficients $x_i$'s except at the point of equal mixing where all $x_i$'s are 1/3 each. From Eq. (\ref{eq:lambdaz}), one can see that the eigenvalues of the output map, $\lambda_i(t) = e^{-ct}$  for the case of equal mixing, which is indicative of a semigroup. For the intermediate case,  i.e., $n \in (4/3,2)$, let us evaluate the measure $\Delta_{\rm invert}$ of invertible maps over all possible mixtures.

To this end, we note that $x_1$ can take values in the range $[1-\frac{n}{2}, 1-2 \times (1-\frac{n}{2})] = [1-\frac{n}{2}, n-1]$, where the lower bound is determined by Eq. (\ref{eq:xi2}), and the upper bound arises due to allocating this minimal mass of probability to the other two variables. The variable $x_2$ then is allowed to take values in the range $1-\frac{n}{2}$ and $\frac{n}{2} - x_1$.  We have
	\begin{eqnarray}
	\Delta_{\rm invert} &=& \frac{1}{\mu}\int_{1-\frac{n}{2}}^{\frac{n}{2}-x_1} dx_2 \int_{1-\frac{n}{2}}^{n-1} dx_1  \nonumber \\
	&=& \frac{(4-3 n)^2}{8},
	\label{eq:Delta}
	\end{eqnarray}
	where the normalization factor $\mu = \int_{0}^{1-x_1}dx_2 \int_{0}^{1} dx_1 = \frac{1}{2}$.
Equation (\ref{eq:Delta}) shows that the measure of invertible maps falls monotonically through the range of $[1,0]$ as $n$ is varied in the intermediate range $(\frac{4}{3},2)$. The above arguments can be readily extended to the case of qudits as performed in Sec.~\ref{genpaulisec}.

\subsection{More general decoherence functions \label{sec:paulisec}}
Going beyond the form of the decoherence function as in Eq.~(\ref{decohfunc}), one can consider other forms of noninvertible maps.  Obviously, there is a great deal of freedom here, and one might wish to start with simple and natural forms. 
We consider the case where each of the noninvertible maps is characterized by the same decoherence function, 
\begin{equation}
\label{decohfunccos}
p(t) = \frac{1-\cos (\omega t )}{2},
\end{equation}
in place of Eq. (\ref{decohfunc}). The corresponding time-local generator can be evaluated to be
 \begin{equation*}
L_{i}(t)[\rho] = \frac{\omega}{2}\tan(\omega t) (\sigma_i\rho\sigma_i-\rho).
\end{equation*}
The decay rate $ \frac{\omega}{2}\tan(\omega t)$ is indicative of the noninvertibility of the input maps.
The output map $\tilde{\Phi}(t)$ Eq. (\ref{outputmappauli}) satisfies the eigenvalue relation
given by Eq. (\ref{eigpauli}), i.e., $\tilde{\Phi}(t) [\sigma_i] = \lambda_i(t) \sigma_i,$ with
\begin{equation}
\lambda_i(t) = x_i + (1-x_i) \cos (\omega t).
\end{equation} 
The eigenvalues $\lambda_i(t)$ become singular at
\begin{equation}
t^{*}=\frac{1}{\omega}\Big[\cos^{-1}\Big(\frac{x_i}{x_i-1}\Big)\Big].
\label{eq:tstar1}
\end{equation}
From this, it can be seen that for the output map to be invertible (i.e., there exists no finite time $t$ at which the eigenvalues can be zero), all three mixing coefficients $x_i,$ should be greater than 0.5, which is impossible. This entails that the resultant map is noninvertible for all choices of convex combinations. A similar behavior is seen for the choice of $p(t) = \sin (\omega t )$.
We note that this generalizes the example in Ref.~\cite{siudzinska_markovian_2021}, where the shift of the singularity in the input maps under convex combination is noted.

Suppose the choice of the decoherence function $p(t)$ in Eq. (\ref{eq:lambdaz}) is such that $p(t)$ monotonically increases up to a finite time $t^\#$, such that $p(t^\#)=\frac{1}{2}$, and it remains constant thereafter. For example:
\begin{equation}
    p(t) = f(t)[1-\Theta(t-t^\#)] + \frac{1}{2}\Theta(t-t^\#),
    \label{eq:heavyp}
\end{equation}
where $f(t)$ is any monotonically increasing function such that $f(0) = 0, f(t^\#)=\frac{1}{2}$. Thus, each input map is maximally dephasing and thereby noninvertible (all ``azimuthal'' points on the Bloch ball being mapped to a single point on the respective axis). Quite generally, Eq. (\ref{eq:lambdaz}) implies that the potential singular point in the output map is the solution $t^*$ for the equation
\begin{equation}
p(t^{*})=\frac{1}{2(1-x_i)}.
\label{eq:tstarp}
\end{equation}
But this is necessarily greater than $\frac{1}{2}$, considering that all $x_i$'s are nonzero. It follows that this defines a set of noninvertible input maps that produces an invertible output.

The above examples in the subsection make the following pattern evident: The decoherence function $p(t)$ can be such that it just makes the input maps noninvertible, but does not increase far enough to make the output map noninvertible; or it can go all the way to 1, making noninvertibility inevitable. We, therefore, understand that the choice of exponential function as in Eq.~(\ref{decohfunc}) is a good candidate to get invertible or noninvertible maps by tuning the parameter $n$ to enable an intermediate behavior. Similar arguments also apply to the case of generalized Pauli maps, discussed below.

\section{Convex combination of noninvertible Generalized Pauli dynamical maps}
\label{genpaulisec}
Based on the concept of {\it mutually unbiased bases} (MUBs), Nathanson and Ruskai  introduced~\cite{ruskai2007} a generalization of the Pauli dynamical maps to the case of maps on qudits. A set of ($d+1$)-orthonormal bases $\{|\xi_i^{(\alpha)}\rangle,i=0,\ldots,d-1\}$ in $\mathbb{C}^d$ such that  $\langle\xi_i^{(\alpha)}|\xi_j^{(\beta)}\rangle|^2=1/d$ whenever $\alpha\neq\beta$, is a MUB. If the dimension of the Hilbert space $d$ is an integer power of a prime number, then $d +1$ MUBs exist. Defining the rank-1 projectors $P_j^{(\alpha)}=|\xi_j^{(\alpha)}\rangle\langle\xi_j^{(\alpha)}|$, one introduces the ($d+1$)-unitary operators,
\begin{equation}
U_\alpha=\sum_{j=0}^{d-1}\omega^jP_j^{(\alpha)},\qquad \omega=e^{2\pi \imath/d}.
\end{equation}
We consider the convex combination of $d+1$-generalized Pauli dynamical maps ($i =1,...d+1)$~\cite{chruscinski2016pra},
\begin{equation}
\Phi_{i}(t) [\rho]=[1-p(t)]\rho + p(t)\sum_{k=1}^{d-1}U_{i}^{k}\rho U_{i}^ {k \thinspace\dagger},
\end{equation}  each mixed in proportions of $x_i$ as
\begin{equation}
\tilde{\Phi}(t) = \sum_{i=1}^{d+1} x_{i} \Phi_i(t),  \quad (x_i >0, \sum_i x_i =1).
\end{equation}
The map $\tilde {\Phi}(t)$ satisfies the eigenvalue relation,
\begin{equation}
\tilde\Phi(t)[U_{i}^{k}] = \lambda_i(t) U_{i}^{k} , k = 1, ..., d-1,
\end{equation}
with
\begin{equation}
\label{eiggenpauli}
\lambda_i(t) = 1-\frac{d}{d-1} (1-x_i)p(t).
\end{equation}
For the choice of $p(t)$ as in Eq.~(\ref{decohfunc}), the eigenvalues become singular at
\begin{equation}
t^{*}=\frac{1}{c}\ln \Big[\frac{d(1-x_i)}{d(1-x_i)-n(d-1)}\Big].
\label{eq:tstargen}
\end{equation}
For the cases of Pauli dynamical maps and generalized Pauli maps, it was shown that when the input map is invertible, the resultant map is necessarily invertible~\cite{siudzinska_markovian_2021,utagi_singularities_2021}. Invertible input maps correspond to $n\ge \frac{d}{d-1}$. From Eq. (\ref{eq:tstargen}), for noninvertible inputs, if $n< \frac{d^2}{d^2-1}$, the resultant map is necessarily noninvertible.   We have the following result.
 \begin{Theorem}
\label{Theorem1}
	 For generalized Pauli input maps that are of the intermediate noninvertible range \big($\frac{d}{d-1} > n \ge \frac{d^2}{d^2-1}$\big), the fraction of output maps that are invertible is given by $\left[\frac{(d^2(n-1)-n)}{d}\right]^d$.
\end{Theorem}
\begin{proof}
	From Eq. (\ref{eq:tstargen}) it follows that the requirement for invertibility is
	\begin{equation}
x_i \ge 1-\frac{n (d-1)}{d} \equiv g(d,n).
\label{eq:xi}
	\end{equation}
The input map is invertible for $n \ge \frac{d}{d-1}$, and is a semigroup for $n = \frac{d}{d-1}$.
For invertible inputs $n\ge \frac{d}{d-1}$, the right-hand side above is negative and, thus, any mixing coefficient $x_i$ satisfies the lower bound in Eq. (\ref{eq:xi}). This bound monotonically increases as $n$ decreases in Eq. (\ref{eq:xi}). If $n < \frac{d^2}{d^2-1}$, then $x_i < \frac{1}{d+1}$, a constraint that cannot be satisfied by all three maps \big(since in that case we would have $\sum_i x_i < 1$\big).

In the intermediate range of noninvertible inputs $n \in (\frac{d^2}{d^2-1},\frac{d}{d-1})$, a fraction of the output maps will be invertible. To evaluate this, first we note that $x_1$ (any given $x_j$) must range from the minimum of $g(d,n)$ in view of Eq. (\ref{eq:xi}), with, at least, this mass of probability being assigned to the remaining $d$ maps. Therefore, the allowed maximum is bounded by $1-dg(d,n)$. Accordingly,
	\begin{equation}
	x_1 \in [g(d,n), 1-d g(d,n)] = [g(d,n), (d-1)(n-1)].
	\end{equation} 
 Similarly, $x_2$ (or, any but $x_1$) must range from the minimum of $g(d,n)$ in view of Eq. (\ref{eq:xi}) with, at least, this mass of probability being assigned to the remaining $(d-1)$ maps less the mass $x_1$ already assigned to the first map \#1. Therefore, the allowed range is given by
 	\begin{equation}
	x_2 \in [g(d,n), 1-(d-1)g(d,n) - x_1].
	\end{equation}
Proceeding analogously, for a general mixing variable $x_j$, we have 
 \begin{equation}
 x_{j+1} \in  [g(d,n), f(j) - X_j],
 \end{equation}
 where $f(j)$ and $X_j$ are as defined by
\begin{align}
 f(j) &\equiv 1-(d-j) g(d,n),\nonumber\\
 X_j &\equiv x_j + x_{j-1} + \cdots + x_1.
 \end{align}
Given dimension $d$ and parameter $n$, the measure of the invertible region in the space of resultant maps is evaluated to be
  \begin{equation}
  \begin{split}
\Delta_{\rm invert}(d,n) &=  \frac{1}{\mu}\int_{g(d,n)}^{(d-1)(n-1)}dx_1\cdots \\  
  &\int_{g(d,n)}^{f(j)-X_j} dx_{j+1} \cdots\int_{g(d,n)}^{f(d-1)-X_{d-1}}dx_d \\
  &=   \left[\frac{d^2(n-1)-n}{d}\right]^d. 
  \end{split}
  \label{eq:Deltagen} 
  \end{equation}
Here, 
  \begin{equation}
  \mu = \int_{0}^{1-X_j}dx_d \int_{0}^{1-X_{j-1}}dx_{d-1}\cdots\int_{0}^{1} dx_1 = \frac{1}{d!}
  \end{equation}is the normalization constant.
\end{proof}
Equation (\ref{eq:Deltagen}) implies that for any given dimension $d$, the invertible fraction $\Delta_{\rm invert}$ ranges from $\Delta_{\rm invert}(d,\frac{d^2}{d^2-1})=0$ to $\Delta_{\rm invert}(d,\frac{d}{d-1})=1$. For fixed $n$, and sufficiently large $d$, we find that $\Delta_{\rm invert}$ approaches unity with a superexponential rate $O(d^d)$. However, to plot this, we remark that the intermediate region $[\frac{d^2}{d^2-1},\frac{d}{d-1}]$ quadratically shrinks into a narrow interval that converges towards unity. Therefore, to depict the dependence of the invertible fraction on $d$, we must choose a range of $d$ values such that a given $n$ value falls within their corresponding interval. An example is given in Fig. \ref{fig:Deltainvert}, varying $d$ discretely over the prime powers from 7 to 32. The intermediate regions of $n$ corresponding to these two values are, respectively, [1.021, 1.167] and [1.001, 1.032]. Here we choose $n=1.03$, which is contained in both intervals and, by interpolation, in the intervals corresponding to all the intermediate values of $d$. (The intervals corresponding to prime powers equal to and below $d=5$ have no overlap with the intervals corresponding to prime powers equal to and above $d=32$, and, hence, are excluded in this representation).
\begin{figure}
	\includegraphics[width=3in]{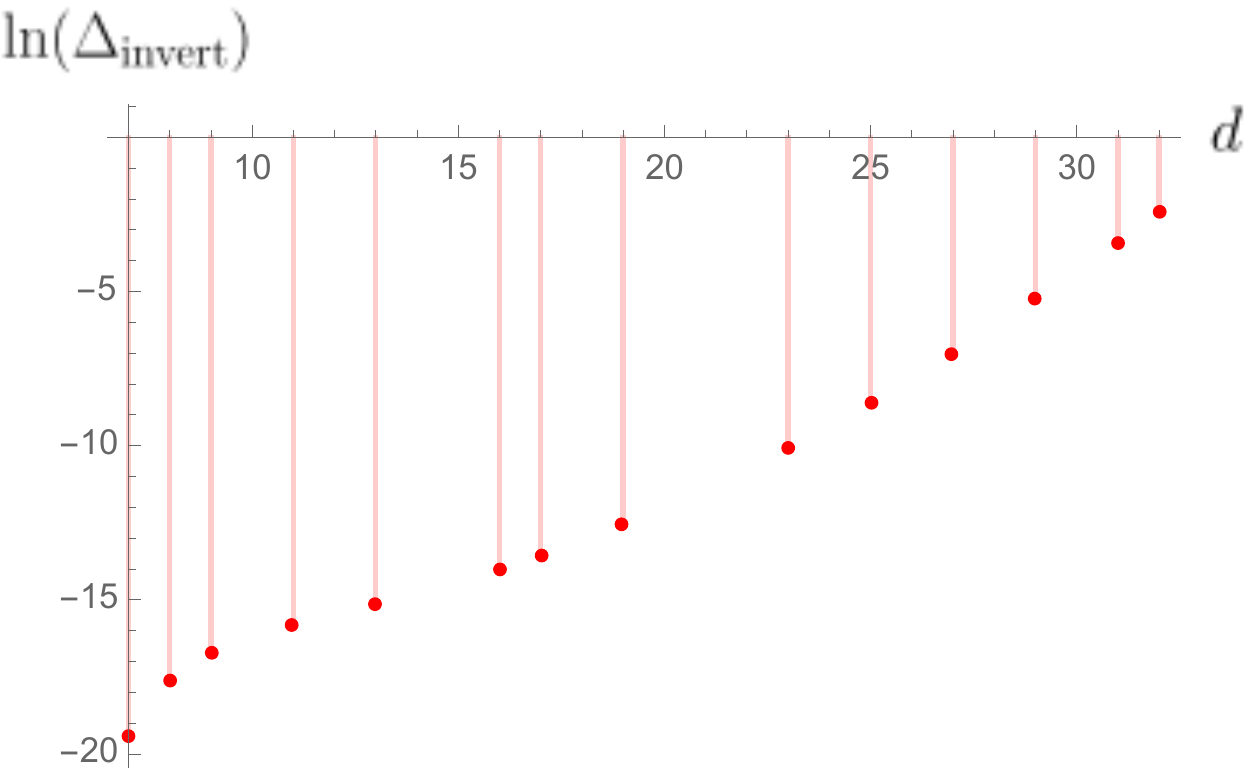}
	\caption{Logarithm plot of the measure of the invertible fraction as a function of $d$ for $n$ {:=} 1.03, which falls within the interval $[\frac{d^2}{d^2-1},\frac{d}{d-1}]$ for $d$ taking prime powers between and including 7 and 32. The plot depicts how the invertible fraction $\Delta_{\rm invert}$ rises at a superexponential rate $O(d^d)$ in the considered parameter space.}
	\label{fig:Deltainvert}
\end{figure}

\section{Discussions and conclusions}
\label{conc}
A quantum resource theory (RT) formalizes the separation of a given set of objects of interest (typically represented as states) into resource states and free (nonresource) states. The idea is that there are free (or ``easy'') operations in the RT such that the set of free states is closed under the action of free operations. For example, in a resource theory of quantum entanglement, the entangled (respectively, separable) states are the resource (respectively, free) states, while local operations and classical communication (LOCC) constitute the free operations. It is well known that separable states cannot become entangled under LOCC. A RT, therefore, addresses the question of tasks which are (im)possible using free operations~\cite{chitambar2019quantum}. In a convex RT, the free states form a convex set. 

One could consider a quantum RT for noninvertibility, considering how this feature can be converted into corresponding limitations on quantum information processing. Consider, for example, a map $\mathcal{E}$ described by the noise $\rho \rightarrow \sum_j K_j \rho K_j^\dagger$, where $K_j$'s are Kraus operators, acting on one-half of the maximally entangled state $\frac{1}{\sqrt{D}}\sum_{k=1}^D\ket{k,k}_{A,B}$ between systems $A$ and $B$. The resulting mixed entangled state can act as a noisy quantum teleportation map. Given an initial state $\rho$ of system $S$, this map teleports $\rho$ according to
\begin{eqnarray}
\rho &\rightarrow& \sum_{rjklm} \bra{l,l}_{SA} \rho_S \otimes K_r \ket{j}_A\bra{k} K_r^\dagger\otimes \ket{j}_B\bra{k} m,m\rangle_{SA}
\nonumber \\
&=& \sum_r K_r^\dagger \rho K_r \equiv \mathcal{E}_T(\rho),
\label{eq:teleportation}
\end{eqnarray}
i.e., a map characterized by Kraus operators $\{K_j\}$ is mapped to one characterized by the Kraus operators $\{K_j^\dagger\}$.
 It follows that if the dynamical map described by $\mathcal{E}$ is noninvertible, then so is the effective teleportation map 
$\mathcal{E}_T$. Thus, we find that quantum noninvertibility as a \textit{negative} resource, i.e., one that is convertible to a limiting (rather than enabling) feature of noisy teleportation, an idea that can be extended to certain other quantum information processing tasks.

In this framework, invertible maps would correspond to free states and noninvertible maps to resource states. Free operations must be resource nonincreasing. For this we consider, the composition of maps, i.e., a sequential application. Invertible maps being bijective, the composition of bijective maps is also a bijection. Since bijective maps have a unique inverse, the composed map, therefore, is also invertible. Thus, the set of invertible maps is closed under composition.  
Given such a RT, our results imply that in the context of qudit Pauli dynamical maps, noninvertibility is a convex resource. Mixing invertible maps can never result in noninvertibility. On the other hand, a mixture of noninvertible maps satisfying the condition $n \in [\frac{d}{d-1}, \frac{d^2}{d^2-1}]$ produces a set of invertible maps with a finite nonzero measure. 
The convexity of the invertible set may be contrasted with the nonconvexity of the set of quantum Markovian (CP divisible) dynamical maps. On that basis, it was argued that CP indivisibility does not yield a convex resource theory~\cite{jagadish_convex_2020}.

In this paper, we have studied the properties of dynamical maps that result under mixing of generalized Pauli dynamical maps that are noninvertible. In any given dimension $d$, (non)invertibility of the input maps is parametrized through a parameter $n$ , which ensures invertibility when sufficiently large. 
Analogously, for noninvertible dynamical maps with sufficiently small $n$, the resultant dynamical map is necessarily noninvertible. We quantify the fraction of invertible maps in the intermediate range of $n$.  Specifically, we have shown that the set of such invertible maps is of nonzero measure. 

\acknowledgements

V.J. acknowledges financial support by the Foundation for Polish Science
through the TEAM-NET Project (Contract No. POIR.04.04.00-00-17C1/18-00). R.S.   acknowledges the support of the Department of Science and Technology (DST), India, Grant No.: MTR/2019/001516.
The work of F.P. is based upon research supported by the South African Research Chair Initiative of the Department of Science and Innovation and National Research Foundation (NRF) (Grant No. UID: 64812).  
\end{document}